\documentclass[a4paper,12pt,nosumlimits]{amsart}       
\title[Solitons on noncommutative spaces]{Sigma-model solitons on noncommutative spaces}               
\date{August 2015}
\author[L.~Dabrowski]{Ludwik Dabrowski}
\address{SISSA (Scuola Internazionale Superiore di Studi Avanzati),
via Bonomea 265, 34136 Trieste, Italy }
\email{dabrow@sissa.it}
\author[G.~Landi]{Giovanni Landi}
\address{Matematica, Universit\`{a} di Trieste, Via A.~Valerio~12/1, I-34127 Trieste, Italy,
and INFN, Sezione di Trieste, Trieste, Italy}
\email{landi@units.it}
\author[F.~Luef]{Franz Luef}
\address{Department of Mathematics, NTNU Trondheim, 7041 Trondheim, Norway}
\email{franz.luef@math.ntnu.no}

\usepackage{dsfont}
\usepackage{amsfonts}
\usepackage{amsmath}
\usepackage{amsthm}
\usepackage{amssymb}
\usepackage{mathrsfs}
\usepackage{latexsym}
\usepackage{float}
\usepackage[top=3cm,right=2.4cm,left=2.4cm,bottom=3cm]{geometry}
\usepackage[colorlinks=true,linkcolor=black,citecolor=black]{hyperref}

\usepackage[usenames,dvipsnames]{color}

\numberwithin{equation}{section}

\begin{document}

\newtheorem{theorem}{Theorem}[section]
\newtheorem{definition}[theorem]{Definition}
\newtheorem{lemma}[theorem]{Lemma}
\newtheorem{proposition}[theorem]{Proposition}
\newtheorem{corollary}[theorem]{Corollary}
\newtheorem{example}[theorem]{Example}
\newtheorem{remark}[theorem]{Remark}
\newcommand{\beq}{\begin{equation}}
\newcommand{\eeq}{\end{equation}}
\newcommand{\nn}{\nonumber}
\newcommand{\dd}{\mathrm{d}}
\newcommand{\mi}{\mathrm{i}}
\newcommand{\ee}{\mathbb{e}}
\newcommand{\IC}{\mathbb{C}}
\newcommand{\IN}{\mathbb{N}}
\newcommand{\IR}{\mathbb{R}}
\newcommand{\IT}{\mathbb{T}}
\newcommand{\IZ}{\mathbb{Z}}
\newcommand{\pa}{\partial}
\newcommand{\opa}{\overline{\partial}}
\newcommand{\onabla}{\overline{\nabla}}
\newcommand{\cE}{\mathcal{E}}
\newcommand{\cH}{\mathcal{H}}
\newcommand{\cP}{\mathcal{P}}
\def\hs#1#2{\left\langle #1,#2\right\rangle}  %
\def\lhs#1#2{{_\bullet\!\!}\left\langle #1,#2\right\rangle}
\def\rhs#1#2{\left\langle #1,#2\right\rangle\!\!{_\bullet}}
\newcommand{\cS}{\mathcal{S}}
\newcommand{\tr}{\mathrm{tr}}
\newcommand{\ac}{\cA_c}

\newcommand{\cAtheta}{\mathcal{A}_\theta}
\newcommand{\cAab}{\mathcal{A}_{\alpha\beta}}

\newcommand{\mcA}{\mathcal{A}}
\newcommand{\cA}{\mathcal{A}}
\newcommand{\cAt}{\mathcal{A}_\theta}
\newcommand{\Alafgra}{{_\mathcal{A}}\langle f,g\rangle}
\newcommand{\lafgraB}{\langle f,g\rangle_\mathcal{B}}

\newcommand{\akl}{a_{kl}}
\newcommand{\amn}{a_{mn}}

\newcommand{\AiLa}{\mathcal{A}^1(\Lambda,c)}
\newcommand{\AisLa}{\mathcal{A}^1_s(\Lambda,c)}
\newcommand{\AiLacirc}{\mathcal{A}^1(\Lambda^\circ,\overline{c})}
\newcommand{\AisLacirc}{\mathcal{A}^1_s(\Lambda^\circ,\overline{c})}
\newcommand{\AinfLa}{\mathcal{A}^\infty(\Lambda,c)}
\newcommand{\AinfsLa}{\mathcal{A}^\infty_s(\Lambda,c)}
\newcommand{\AinfLacirc}{\mathcal{A}^\infty(\Lambda^\circ,\overline{c})}
\newcommand{\AinfsLacirc}{\mathcal{A}^\infty_s(\Lambda^\circ,\overline{c})}
\newcommand{\Ata}{\mathcal{A}_\theta}
\newcommand{\aZZ}{\alpha\mathbb{Z}}
\newcommand{\bZZ}{\beta\mathbb{Z}}
\newcommand{\aZbZ}{\alpha\mathbb{Z}\times\beta\mathbb{Z}}
\newcommand{\thZZ}{\theta\mathbb{Z}\times\mathbb{Z}}
\newcommand{\ZthZ}{\mathbb{Z}\times\theta^{-1}\mathbb{Z}}

\newcommand{\mcB}{\mathcal{B}}
\newcommand{\cB}{\mathcal{B}}

\newcommand{\bkl}{b_{kl}}
\newcommand{\bmn}{b_{mn}}

\newcommand{\la}{\langle}
\newcommand{\ra}{\rangle}
\newcommand{\mcC}{\mathcal{C}}
\newcommand{\CC}{\mathbb{C}}
\newcommand{\CmS}{\mathrm{C}^*}
\newcommand{\CstLatw}{\mathrm{C}^*(\Lambda,c)}
\newcommand{\CstLacirctw}{\mathrm{C}^*(\Lambda^\circ,\overline{c})}
\newcommand{\cAVcB}{{_\mathcal{A}}{V}_\mathcal{B}}

\newcommand{\GabFrame}{A\|f\|_2^2\le\sum_{\lambda\in\Lambda}|\langle f,\pi(\lambda)g\rangle|^2\le B\|f\|_2^2}
\newcommand{\GmcgLa}{\mathcal{G}(g,\Lambda)}
\newcommand{\Ggab}{\mathcal{G}(g,\alpha\ZZ\times\beta\ZZ)}

\newcommand{\cK}{\mathcal{K}}
\newcommand{\livsp}{\ell^1_v}
\newcommand{\LtR}{L^2(\mathbb{R})}
\newcommand{\LtRR}{L^2(\mathbb{R})}

\newcommand{\piakbl}{\pi(\alpha k,\beta l)}
\newcommand{\pibkal}{\pi(\tfrac{k}{\beta},\tfrac{l}{\alpha})}
\newcommand{\piab}{\pi(\alpha k),\beta m)}
\newcommand{\piabg}{\pi(\alpha k),\beta m)}
\newcommand{\pithkl}{\pi(\theta k,l)}
\newcommand{\pikthl}{\pi(k,\theta^{-1}l)}

\newcommand{\R}{\mathbb{R}}
\newcommand{\RR}{\mathbb{R}}
\newcommand{\RRt}{{\mathbb{R}^2}}

\newcommand{\SR}{\cS(\R)}
\newcommand{\SRt}{\cS(\RRt)}

\newcommand{\Sgab}{S_{g}^{\alpha,\beta}}

\newcommand{\trt}{\mathrm{tr}_\Lambda}
\newcommand{\TT}{\mathbb{T}}
\newcommand\WLinfliv{{ \Wsp(\Linsp,\livsp) }}   
\newcommand{\Z}{\mathbb{Z}}
\newcommand{\ZZ}{\mathbb{Z}}
\newcommand{\ZZt}{\mathbb{Z}^2}

\def\Xint#1{\mathchoice
{\XXint\displaystyle\textstyle{#1}}%
{\XXint\textstyle\scriptstyle{#1}}%
{\XXint\scriptstyle\scriptscriptstyle{#1}}%
{\XXint\scriptscriptstyle\scriptscriptstyle{#1}}%
\!\int}
\def\XXint#1#2#3{{\setbox0=\hbox{$#1{#2#3}{\int}$ }
\vcenter{\hbox{$#2#3$ }}\kern-.6\wd0}}
\def\ddashint{\Xint=}
\def\hint{\Xint-}

\subjclass[2010]{Primary: 58E20, 35C08; Secondary: 42C15, 58B34 42B35}

\begin{abstract}
We use results from time-frequency analysis and Gabor analysis to construct new classes of sigma-model 
solitons over the Moyal plane and over noncommutative tori, taken as source spaces, with a target space made of two points. A natural action functional leads to self-duality equations for projections in the source algebra. Solutions, having non-trivial topological content, are constructed via suitable Morita duality bimodules.  
\end{abstract}

\maketitle

\vfill
\tableofcontents

\parskip 1ex
\linespread{1.2}

\vfill
\newpage

\section{Introduction}

Noncommutative analogues of non-linear sigma-models or, said otherwise, noncommutative harmonic maps were introduced in \cite{DKL00,DKL03}. Already the simplest example of a target space made of two points --- for which the commutative theory is trivial --- showed remarkable properties when taken together with a noncommutative source space. A natural action functional led to self-duality equations for projections in the source algebra, with solutions (on noncommutative tori) having non-trivial topological content (topological charges). Some generalizations were presented in \cite{MR11} and further results in \cite{le14}. 

The constructions of \cite{DKL00,DKL03} relied on complex structures and the existence of holomorphic structures for projective modules over noncommutative spaces. The projections in the source algebra were coming from a 
Morita equivalence bimodule. 

Now, an important instance of the Morita equivalence occurs between the algebra of compact operators on a separable Hilbert space and the algebra $\IC$, a result which is equivalent \cite{Ri72} to the uniqueness of the Schr\"odinger representation of the CCR (for the two-dimensional quantum phase space of a free particle). When thinking of a space of positions rather than of the phase space (an interpretation we favor for the discussion of sigma models), the same algebra is now the algebra of the so called Moyal plane.

In turn, when passing from the Schr\"odinger representation of phase space to that of lattices in phase space, 
the same construction of Morita equivalence bimodules over the Moyal plane leads to equivalence bimodules over 
noncommutative tori whose noncommutativity parameters are related by a fractional transformation with integer coefficients \cite{Ri81}.
At the smooth level, the bimodule that implements both instances of Morita equivalence for the lowest value of `the rank' is the underlying vector space $\SR$ of Schwartz functions on $\IR$.

Most remarkably, these same objects appear in the field of time-frequency 
analysis and Gabor analysis \cite{ga46}. In the light of the relation between these and noncommutative geometry established in \cite{lu09,lu11}, 
the present paper is centered on the use of several results from time-frequency analysis and Gabor analysis to construct new classes of sigma-model solitons over the Moyal plane and over noncommutative tori. 

In the setting of noncommutative 
geometry the main goal of time-frequency and Gabor analysis is to find generators for projective modules over noncommutative tori. 
It is these results that will provide new classes of solitons over noncommutative tori. 

In \S\ref{se:afse} 
we recall the noncommutative sigma model of maps from a virtual space underlying a noncommutative algebra $\cA$
to the set of two points, as given in terms of projections in the algebra $\cA$.
The dynamics is governed by a quadratic action functional, defined in terms of suitable derivations and an invariant trace on $\cA$.  
The functional is shown to be estimated from below by a {\it topological charge} (a Chern number), with extremal points which are solutions of first order differential self-duality equations for projections. 

In \S\ref{se:eqbi} we give some basic material on smooth Morita equivalence between a pair
of algebras $\cA$ and $\cB$ via an equivalence $\cA-\cB$ bimodule $\cE$.
We spell out our requirements on the compatible actions of a pair of derivations 
and of their lift (implementation) on the module $\cE$ in terms of a pair of covariant derivatives (a connection).
The  self-duality equation for projections in the algebra $\cA$ of \S\ref{se:afse} is given as  
a generalized eigenfunction problem for an anti-holomorphic connection $\onabla$ on the module $\cE$ 
with eigenfunctions in the algebra $\cB$. The topological charge, which measures the nontriviality of the 
module $\cE$, is expressed in terms of the curvature of the connection on $\cE$. 

Results from the Schr\"odinger representation that are needed in the later part  are in \S\ref{se:sr}. 

For the Moyal plane in \S\ref{se:moyp}, the general solutions of the generalized eigenfunction problem for the operator $\onabla$ are given by Gaussians functions. Now the (constant) curvature of the connection is interpreted as Heisenberg commutation relations and the self-duality equation for the corresponding (Gaussian) projections is the equation for the minimizers of the Heisenberg uncertainty relation.
More generally, for noncommutative tori in \S\ref{se:nct}, solutions of the generalized eigenfunction problem for $\onabla$ are obtained from 
Gabor frames: these solutions include Hermite functions and totally positive functions of finite type. It is remarkable that 
for these solution we rely on two cornerstones of the field of Gabor analysis, namely the duality theory and 
the Wexler--Raz identity \cite{dalala95,ja95,rosh97,rawe90}.

\subsubsection*{Acknowledgments.}
Part of the work was carried out during visits in Vienna at the Erwin Schr\"odinger Institute for Mathematical Physics and in Bonn at the Hausdorff Research Institute for Mathematics. We thank the organizers of the Programs for the invitation and all people at ESI and HIM for the nice hospitality. 
This work was partially supported by the Italian Project       
`Prin 2010-11 -- Operator Algebras, Noncommutative Geometry and Applications, and by the GNSAGA of the Italian `Istituto Nazionale di Alta  Matematica' (INdAM).  L.D. acknowledges a partial support
from HARMONIA NCN grant 2012/06/M/ST1/00169.

\section{The action functional and the soliton equations}\label{se:afse}

Noncommutative analogues of non-linear sigma-models were constructed in  \cite{DKL00,DKL03}. 
The simplest example was with a target space made of two points $M = \{1,2\}$.
Since any continuous map from a connected surface $\Sigma$ to a discrete space is constant, a commutative theory would be trivial. 
This is not the case if the source space is `noncommutative' 
and there are, in general, nontrivial such maps. Now they have to be interpreted `dually', i.e. as $*$-algebra morphisms from the algebra of functions over $M = \{1,2\}$, that is $\IC^{2}$, to the algebra
$A$ of the noncommutative source space.
Since as a vector space $\IC^{2}$ is generated by the projection (self-adjoint idempotent) function $e$ defined by $e(1)=1$ and $e(2)=0$, any $*$-algebra morphism $\pi : \IC^{2} \to A$ is identified with a projection $p=\pi(e)$ in $A$. As a consequence the configuration space of a two point target space sigma-model is the  collection of all projections $\cP(A)$ in the algebra $A$.

For the dynamics of the model we need additional structures. 
At the continuous level, we take the $C^*$-algebra $A$ to carry an 
action of the torus group $\IT^2$ with $\cA$ denoting the dense pre $C^*$-algebra of corresponding smooth elements, and infinitesimal generators of the action denoted $\pa_1$ and $\pa_2$: these are derivations of $\cA$. 
Also, we assume $\cA$ to have an invariant faithful tracial state $\tr$; the invariance means that $\tr(\pa_1(a))=\tr(\pa_2(a))=0$ for all $a\in\cA$. 
Then, on the configuration space $\cP(\cA)$ we consider the $\IR^+$-valued action-functional
\beq\label{actfun}
S[p] = \frac{1}{4 \pi} \tr \big( (\pa_1 p)^2 + (\pa_2 p)^2 \big) .
\eeq
We stress that this action functional depends on the choice of a metric. In the present situation we work with the flat one,  
with respect to which the derivatives are orthogonal and normalized.
Using the identity $p^2=p$ and the trace property, one has also
$$ 
S[p] = \frac{1}{2 \pi} \tr \, p \left[ (\pa_1 p)^2 + (\pa_2 p)^2 \right] .
$$
Moreover, using the natural complex structure on $\cA$ given by
$$ 
\pa = \pa_1 - \mi \, \pa_2, \qquad \opa = \pa_1 + \mi \, \pa_2 ,
$$
the action functional can be written as
$$ 
S[p] = \frac{1}{4 \pi} \tr ( \pa p \opa p ) .
$$

As usual, the critical points of the action functional \eqref{actfun} are obtained by equating to zero its first
variation, that is the linear term in an infinitesimal variation $\delta S[p]$ = $S[p+\delta p] - S[p]$
for $\delta p \in T_p(\cP(\cA))$.
By `integrating by parts' and using the invariance of the integral to get rid of the `boundary  terms',
and the trace property, one easily gets:
\beq\label{first}
0 = \delta S[p] = - \frac{1}{2 \pi} \tr \, \Delta(p) \, \delta p ,
\eeq
where $\Delta = \frac{1}{2} (\pa \opa + \opa \pa) = \pa_1^2 + \pa_2^2$,
is the Laplacian of the flat metric.
On the other hand, the most general element $\delta p \in T_p(\cP(\cA))$ is not arbitrary but rather of the form
\beq\label{tanpro}
\delta p = (1-p) z p + p z^* (1-p),
\eeq
with $z$ an arbitrary elements in $\cA$. This is because, starting from a general expression
$$
\delta p = p x p + (1-p) y (1-p) + (1-p) z p + p z^* (1-p),
$$
the condition that this expression is an idempotent to first order, $(p + \delta p)^2 = p + \delta p + O(\delta p)$, forces $x=y=0$, while the additional condition that it is hermitian, $(\delta p)^* = \delta p$, forces $w=z^*$.
When substituting \eqref{tanpro} into \eqref{first}, using again the trace property one arrives at
\begin{equation*}
0 = - \frac{1}{4 \pi} \tr \big( \left[ p ~\Delta(p)~(1-p)\right] z  + \left[ (1-p) ~\Delta(p) ~p \right]z^* \big) .
\end{equation*}
Since $z$ is arbitrary one finally gets the equations for the critical points:
\begin{equation*}
p ~\Delta(p) ~(1-p) = 0 \qquad {\rm and} \qquad (1-p) ~\Delta(p) ~p = 0,
\end{equation*}
or, equivalently, the following non-linear equations of the second order
\beq\label{eom}
p ~\Delta(p) - \Delta(p) ~p = 0.
\eeq

As mentioned, the action functional in \eqref{actfun} depends on a metric. 
On the other hand, there is a {\em topological charge}, 
that is a quantity not depending on the metric, given by  
\beq
\label{topch}
c_1(p):=
\frac{1}{2\pi \mi} \, \tr \big( p (\pa_1 p \pa_2 p -\pa_2 p \pa_1 p ) \big) .
\eeq
The normalization in \eqref{topch} is such that, in all cases of interest of the present paper, for all projections $c_1(p)$ it is an integer.  It is in fact the index of a Fredholm operator.

A remarkable fact is that in a component of the space $\cP(\cA)$ with a fixed value of $c_1(p)$ the equations \eqref{eom} lead to first order equations as
a consequence of the fact that the topological quantity $c_1(p)$ is a lower bound for the action functional $S[p]$. 
\begin{proposition}
For any $p \in \cP(\cA) $ it holds that 
\beq\label{bpbou}
S[p] \geq |c_1(p)| .
\eeq
\end{proposition}
\begin{proof} 
Due to positivity of the trace integral and its cyclic properties, one has that
\begin{align*}
0 & \leq \tr \big( [ (\pa_1 \pm \mi \, \pa_2)(p) \, p ]^* [ (\pa_1 \pm \mi \, \pa_2)(p) \, p ] \big) \\
& = \tr \big( p [ (\pa_1 p)^2 + (\pa_2 p)^2 ] \big) \pm \mi \, \tr \big( p [ \pa_1 p \, \pa_2 p  - \pa_2 p \, \pa_1 p ] \big) ,
\end{align*}
from which it follows that
\begin{equation*}
\tr \big( p [ (\pa_1 p)^2 + (\pa_2 p)^2 ] \big) \geq \left| \tr \big( p [ \pa_1 p \, \pa_2 p  - \pa_2 p \, \pa_1 p ] \big) \right| ,
\end{equation*}
and  comparing with the definitions \eqref{topch} and \eqref{actfun} one gets the inequality \eqref{bpbou}.
\end{proof}

From the proof above, it is clear that the equality in \eqref{bpbou} occurs when the projection 
$p$ satisfies {\it self-duality} or {\it anti-self duality} equations:
$$ 
(\pa_1 \pm \mi \, \pa_2)(p) \, p= 0 ,
$$
or equivalently $\, p (\pa_1 \pm \mi \, \pa_2)(p) = 0$. These can also be written respectively as
\beq\label{sd-asd}
\opa(p) \, p =0  , \qquad \pa(p) \, p =0  ,
\eeq
or equivalently $p\, \pa(p) = 0 $, respectively $p\, \opa(p) = 0$.

All of the above can be extended to more general metrics, see \cite{DKL00}. In two dimensions, the conformal class of a general constant metric is parametrized by
a  complex number $\tau \in \IC$, with $\Im \tau > 0$.
Up to a conformal factor, the metric is
given by
$$ 
g = (g_{\mu\nu}) =
\left(
\begin{array}{cc}
1 & \Re\tau \\
\Re\tau & |\tau|^2
\end{array}
\right) .
$$
The corresponding `complex torus' $\IT^2$ would act on $\cA$ infinitesimally with two derivations
$$ 
\pa = \pa_1 +\bar{\tau} \pa_2 , \qquad
\opa = \pa_1 + \tau \pa_2 .
$$

\section{Equivalence Bimodules}\label{se:eqbi}

The construction of projections in a $C^*$-algebra $A$ via an equivalence bimodule $E$ with a second unital $C^*$-algebra $B$ was a crucial result of Rieffel \cite{Ri81,ri88}.  Here we recall the main results of the construction in the smooth version.

\subsection{Projections from bimodules}\label{prbim}
  
With pre $C^*$-algebras $\cA$ and  
$\cB$, an equivalence $\cA-\cB$-bimodule $\cE$ between them is equipped with a left-linear $\cA$-valued hermitian product on $\cE$, that we denote by $\lhs{\cdot}{\cdot}$, and a right-linear $\cB$-valued hermitian product on $\cE$ denoted by $\rhs{\cdot}{\cdot}$. Thus, $\cE$ is both a left and a right (pre-)Hilbert module. In addition, the hermitian products satisfy 
an associativity condition:  
\beq\label{ascond}
\lhs{\xi}{\eta} \zeta = \xi \rhs{\eta}{\zeta} ,
\eeq
for all $\xi,\eta,\zeta\in \cE$. When there exists such an equivalence bimodule $\cE$ between $\cA$ and $\cB$ one says that the two algebras are Morita equivalent \cite{ri74-2}. 
 Morita equivalence yields an identification 
of the right algebra with the compact endomorphisms of $\cE$, $\cB \simeq \mathrm{End}_{\cA}^0(\cE)$. 
In particular, when the algebra $\cB$ is unital there exist elements $\{ \eta_1,...,\eta_n \}$ in $\cE$ such that 
$$
\sum_j \rhs{\eta_j}{\eta_j} = 1_\cB \,. 
$$
As a consequence, the matrix 
$p = (p_{jk})$ with elements $p_{jk} = \lhs{\eta_j}{\eta_k}$ is a projection in the matrix algebra $M_n(\cA)$:
$$
(p^2)_{jl} = \sum_k \lhs{\eta_j}{\eta_k} \lhs{\eta_k}{\eta_l} =  \sum_k \lhs{\lhs{\eta_j}{\eta_k} \eta_k}{\eta_l}  
= \sum_k \lhs{ \eta_j \rhs{\eta_k}{\eta_k} }{\eta_l} = p_{jl} ,
$$
having used the associativity condition \eqref{ascond}. This establishes the finite left $\cA$-module projectivity of $\cE$ with  the identification $\cE \simeq \cA^n p$.  Furthermore, an use of the condition \eqref{ascond} allows one to reconstruct any element 
$\xi\in\cE$ over the family $\{ \eta_1,...,\eta_n \}$:  
\begin{align}\label{recon}
  \xi &= \xi \, 1_\cB  = \xi \, \sum_j \rhs{\eta_j}{\eta_j}  \nonumber \\
       &=\lhs{\xi}{\eta_1} \eta_1+\cdots+\lhs{\xi}{\eta_n} \eta_n .
\end{align}
By results of \cite{frla02}, this is rephrased as the existence of a {\it Parseval standard module frame} $\{\eta_1,...,\eta_n\}$ for $\cE$. 
In general, a standard module frame for $\cE$ is a set $\{\eta_1,...,\eta_n\}$ such that
\beq\label{frame}
  c_1\,\lhs{\xi}{\xi}\le\sum_{j} \lhs{\xi}{\eta_i}\lhs{\eta_i}{\xi}\le c_2\,\lhs{\xi}{\xi} \,, \qquad  \textup{for all} \quad \xi\in\cE ,
\eeq
for positive constants $c_1$ and $c_2$. The frame $\{\eta_1,...,\eta_n\}$ 
is said to be tight if $c_1=c_2$, and to be normalized or a Parseval frame if $c_1=c_2=1$.
Thus, if $\cA$ and $\cB$ are Morita equivalent, there exists a Parseval standard module frame for the $\cA-\cB$ equivalence bimodule $\cE$. However, it will be useful for us to consider (and start with) more general standard module frames.  

The module $\cE$ is self-dual for the $\cA$-valued hermitian structure \cite[Prop.~7.3]{ri10-2}, in the sense that for any 
$\varphi \in {}_{\cA}\mathrm{Hom}({}_{\cA}\cE,{}_{\cA}\cA)$ there exists a unique $\zeta_\varphi\in \cE$ such that 
$\varphi(\xi)=\lhs{\xi}{\zeta_\varphi}$, for all $\xi\in\cE$. Indeed,
$$
\varphi(\xi) = \varphi \Big( \sum_j \lhs{\xi}{\eta_j} \eta_j \Big) = \sum_j \lhs{\xi}{\eta_j} \varphi(\eta_j ) = 
\lhs{\xi}{ \sum_j (\varphi(\eta_j))^* \eta_j}  ,
$$
and $\zeta_\varphi = \sum_j (\varphi(\eta_j))^* \eta_j$ is the element of $\cE$ representing $\varphi$. 
Thus every element of ${}_{\cA}\mathrm{Hom}({}_{\cA}\cE,{}_{\cA}\cA)$ can be written as $\varphi_\xi$ for some $\xi\in\cE$.
By \cite[Prop.~7.3]{ri10-2}, starting with a left $\cA$-module $\cE$ which is self-dual, any projection $p$ such that $\cE \simeq \cA^n p$ is of the form $p_{jk} = \lhs{\eta_j}{\eta_k}$ for some Parseval standard module frame $\{\eta_1,...,\eta_n\}$ for $\cE$. 

There is in fact more structure. Firstly, from the above discussion, one has a linear space identification ${}_{\cA}\mathrm{Hom}({}_{\cA}\cE,{}_{\cA}\cA)\simeq \cE$:
for $\xi\in \cE$, the corresponding $\varphi_\xi \in {}_{\cA}\mathrm{Hom}({}_{\cA}\cE,{}_{\cA}\cA)$ is given by $\varphi_\xi(\eta) = \lhs{\eta}{\xi}$, for any $\eta\in\cE$.  
 Next, a {\it right} action of $a\in\cA$ is defined as 
$$ 
\varphi_\xi \cdot a = R_a \circ  \varphi_\xi = \varphi_{a^* \xi}, \quad \textup{that is} \quad (\varphi_\xi \cdot a )(\psi) := \lhs{\psi}{\xi} a ,
$$ 
while a {\it left} action of $b\in\cB$ is defined by 
$$ 
b \cdot \varphi_\xi = \varphi_\xi \circ R_b , \quad \textup{that is} \quad ( b \cdot \varphi_\xi )(\psi) 
:= \lhs{\psi b}{\xi} . 
$$
These allow one to reconstruct any element $\varphi_\xi$ over the frame $\{ \eta_1,...,\eta_n \}$. From \eqref{recon}:
\begin{align*}
 \varphi_\xi (\psi) & = ( \varphi_{ \sum_j \lhs{\xi}{\eta_j} \eta_j} )(\psi)   
 = \sum_j \lhs{\psi}{ \lhs{\xi}{\eta_j} \eta_j } 
 \\ 
 & = \sum_j \lhs{\psi}{\eta_j} \lhs{\eta_j }{\xi} = \sum_j \varphi_{\eta_j}(\psi) \lhs{\eta_j }{\xi} 
 = \sum_j ( \varphi_{\eta_j} \cdot \lhs{\eta_j }{\xi}) (\psi) ,
\end{align*}
having used the properties of the left hermitian structure. Equivalently, one can compute:
\begin{align*}
 \varphi_\xi (\psi) & = ( 1_\cB \varphi_\xi )(\psi) = \sum_j \big( \rhs{\eta_j}{\eta_j} \varphi_\xi \big) (\psi) 
 = \sum_j \lhs{ \psi \rhs{\eta_j}{\eta_j} }{\xi} = \sum_j \lhs{ \lhs{\psi}{\eta_j} \eta_j }{\xi} 
 \\ 
 & = \sum_j \lhs{\psi}{\eta_j} \lhs{\eta_j }{\xi} = \sum_j \varphi_{\eta_j}(\psi) \lhs{\eta_j }{\xi} 
 = \sum_j ( \varphi_{\eta_j} \cdot \lhs{\eta_j }{\xi}) (\psi) ,
\end{align*}
having used the associativity condition \eqref{ascond}, and arriving at the same result. The expression above is a reconstruction formula which is dual to the one in \eqref{recon}, that is,
$$
 \varphi_\xi = \varphi_{\eta_1} \cdot \lhs{\eta_1 }{\xi} +\cdots+ \varphi_{\eta_n} \cdot \lhs{\eta_n }{\xi} .
$$

\begin{remark}
\textup{
In more generality one could start with dual frames $\{ \eta_j \}$  and $\{ \zeta_j \}$ in $\cE$ for which
$\sum_j \rhs{\eta_j}{\zeta_j} = 1_\cB$. The matrix $e = (e_{jk})$ in $M_n(\cA)$ with elements $e_{jk} = \lhs{\zeta_j}{\eta_k}$ is now 
only an idempotent $e^2=e$, and any element $\xi\in\cE$ can be reconstructed as
$$ 
  \xi 
  =\lhs{\xi}{\eta_1} \zeta_1+\cdots+\lhs{\xi}{\eta_n} \zeta_n ,
$$ 
while elements $\varphi_\xi \in {}_{\cA}\mathrm{Hom}({}_{\cA}\cE,{}_{\cA}\cA)$ are reconstructed as 
$$ 
  \varphi_\xi 
  = \varphi_{\eta_1} \cdot \lhs{\zeta_1 }{\xi} +\cdots+ \varphi_{\eta_n} \cdot \lhs{\zeta_n }{\xi} .
$$
}
\end{remark}
In this paper we address the case when the $\cA-\cB$ equivalence bimodule $\cE$ has one generator $\eta$, while 
postponing to a future paper the general situation of an arbitrary finitely generated equivalence bimodule $\cE$. 
Hence, $\eta$ is a Parseval frame for $\cE$ if it holds that 
\begin{equation*}  
  \xi=\lhs {\xi}{\eta}\eta=\xi\rhs{\eta}{\eta}
\end{equation*}
for all $\xi\in\cE$, 
that is $\eta$ is a Parseval frame for $\cE$ if and only if $\rhs{\eta}{\eta}=1_\cB$.

If one starts with a standard module frame $\eta$, the element $\rhs{\eta}{\eta}$ is invertible and positive \cite{frla02} and one gets a Parseval frame $\tilde{\eta}$ by considering the element $\tilde{\eta}:=\eta(\rhs{\eta}{\eta})^{-1/2}$.

We finish the section recalling a slightly more general result \cite[Prop.~2.8]{Ri81}:

\begin{lemma}\label{Rieffel}
Suppose $\cE$ is an equivalence $\cA-\cB$-bimodule and let $\psi\in \cE$.
Then $p=\lhs{\psi}{\psi}$ is a projection in $\cA$ if and only if
$$
\psi\rhs{\psi}{\psi}=\psi .
$$
In particular, any Parseval frame $\psi\in \cE$, $\rhs{\psi}{\psi}=1_\cB$, yields a projection $p=\lhs{\psi}{\psi}  \in \cA$. 
\end{lemma}
\begin{proof}
Suppose $\psi\rhs{\psi}{\psi}=\psi$. The associativity condition \eqref{ascond}
and left-linearity yield:
\begin{eqnarray*}
\lhs{\psi}{\psi} \lhs{\psi}{\psi} =  \lhs{\lhs{\psi}{\psi} \psi}{\psi} = \lhs{\psi \rhs{\psi}{\psi} }{\psi} = \lhs{\psi}{\psi}  ,
\end{eqnarray*}
that is $\lhs{\psi}{\psi}$ is a projection. Conversely, suppose $\lhs{\psi}{\psi}$ is a projection in $\cA$. Then a simple calculation using once more the associativity condition \eqref{ascond} implies that
\begin{equation*}
\lhs{\psi \rhs{\psi}{\psi} - \psi} {\psi \rhs{\psi}{\psi} - \psi}=0 ,
\end{equation*}
and thus one obtains $\psi \rhs{\psi}{\psi} = \psi$.   
\end{proof}
 
\subsection{Derivations and connections}
We assume next that both algebras $\cA$ and $\cB$ are contained in the joint smooth domain of
two commuting derivations $\pa_1$ and $\pa_2$ (denoted by the same symbols on the two algebras). 
Moreover $\cA$ and $\cB$ are endowed with faithful tracial states, denoted by the same symbol $\tr$. 
We also  require that the traces are invariant, 
that is $\tr (\pa_j(a))=0$, for all $a\in \cA$ and $\tr (\pa_j(b))=0$, for all $b\in \cB$. 
In addition, we assume the traces to be compatible in the sense that
$$
 \tr \; \lhs{\xi}{\eta} =  \tr \rhs{\eta}{\xi} \quad \text{for all} \quad \xi, \eta \in \cE .
$$

Finally, we assume that the derivations can be consistently lifted to the bimodule $\cE$ as two covariant
derivatives $\nabla_1$ and $\nabla_2$. That is, there are linear maps
\beq
\nabla_j : \cE \to \cE , \quad j=1,2 ,
\eeq
which satisfy a left and right Leibniz rule:
for all $\xi\in \cE$, $a\in \cA$ and $b\in \cB$ it holds that
$$ 
\nabla_j (a\, \xi) = (\pa_j a) \, \xi\  + a\, (\nabla_j \xi) , \quad \textup{and} \quad 
\nabla_j (\xi\, b) = (\nabla_j \xi)\, b +  \xi (\pa_j b) .
$$ 
The covariant derivatives are taken to be compatible with both the
$\cA$-valued  hermitian structure $\lhs{\cdot}{\cdot}$ and the
$\cB$-valued  hermitian structure $\rhs{\cdot}{\cdot}$, that is for all $\xi, \eta\in \cE$, both 
\beq\label{lcomcov}
\pa_j(\lhs{\xi}{\eta}) = \lhs{\nabla_j \xi}{\eta} + \lhs{\xi}{\nabla_j \eta}
\eeq
and
\beq\label{rcomcov}
\pa_j(\rhs{\xi}{\eta}) = \rhs{\nabla_j \xi}{\eta} + \rhs{\xi}{\nabla_j \eta},
\eeq
hold. The curvature of the covariant derivatives is defined as
$$
F_{12} := \nabla_1\,\nabla_2\ - \nabla_2\nabla_1
$$
and it is easily seen to be both left $\cA$-linear and right $\cB$-linear. 

Recall the topological charge $c_1(p)$ defined in \eqref{topch} for any projection. 
\begin{proposition}\label{tcF}
Let $\psi \in \cE$ be such that $\rhs{\psi}{\psi}=1_{\cB}$ and
$p_\psi := \lhs{\psi}{\psi}\in\cA$ the corresponding projection. Then, for its topological charge one finds
$$
c_1(p_\psi) = - \frac{1}{2\pi \mi} \, \tr \rhs{\psi}{F_{12}\psi}  . 
$$
\end{proposition}
\begin{proof}
With a projection $p_\psi := \lhs{\psi}{\psi}$ for $\psi \in \cE$ such that $\rhs{\psi}{\psi}=1_{\cB}$, a direct computation, using the 
compatibility \eqref{lcomcov} and the associativity condition \eqref{ascond}, leads to
\begin{multline*}
p_\psi(\pa_1 p_\psi \pa_2 p_\psi -\pa_2 p_\psi \pa_1 p_\psi) = 
\lhs{\psi\,\left(\rhs{\nabla_1\psi}{\nabla_2\psi} - \rhs{\nabla_2\psi}{\nabla_1\psi}\right)}{\psi} \\
+ \lhs{ \psi \,\left(\rhs{\psi}{\nabla_1\psi}\rhs{\psi}{\nabla_2\psi} - \rhs{\psi}{\nabla_2\psi}\rhs{\psi}{\nabla_1\psi} \right)}{\psi} .
\end{multline*}
Next, the compatibility of traces and the tracial property, with $\rhs{\psi}{\psi}=1_{\cB}$, yield
$$
\tr\big( (p_\psi(\pa_1 p_\psi \pa_2 p_\psi -\pa_2 p_\psi \pa_1 p_\psi) \big)  = \tr \big( \rhs{\nabla_1\psi}{\nabla_2\psi}
- \rhs{\nabla_2\psi}{\nabla_1\psi}\big).
$$
In turn, the compatibility \eqref{rcomcov} and the invariance of the trace leads to 
$$
\tr\big( (p_\psi(\pa_1 p_\psi \pa_2 p_\psi -\pa_2 p_\psi \pa_1 p_\psi) \big) = - \tr \rhs{\psi}{F_{12}\psi}
$$
which completes the proof.
\end{proof}
\begin{corollary}\label{tcconst}
Let the curvature be constant and equal to $F_{12} = -2\pi \mi \,q\, \mathrm{id_\cE}$.
Then, for any $\psi \in \cE$ such that $\rhs{\psi}{\psi}=1_{\cB}$, 
the projection $p_{\psi}=\lhs{\psi}{\psi}$ has topological charge
$$
c_1(p) = q\, \tr(1_B) .
$$
\end{corollary}
 
\subsection{Soliton solutions}

We look for solutions of the self-duality equations \eqref{sd-asd} of the form
$p_\psi := \lhs{\psi}{\psi}\in\cA$ with $\psi \in \cE$ such that $\rhs{\psi}{\psi}=1_{\cB}$ as in Lemma~\eqref{Rieffel}.
We shall need the holomorphic, respectively anti-holomorphic, connection on $\cE$,
$$
\nabla = \nabla_1 -\mi \,  \nabla_2, \qquad
\onabla =\nabla_1 +\mi \,  \nabla_2 ,
$$
which lift to $\cE$ the corresponding complex derivative $\pa = \pa_1 - \mi \, \pa_2$ or $\opa = \pa_1 + \mi \, \pa_2$. 

\begin{proposition}\label{main0}
Let $\psi \in \cE$ be such that $\rhs{\psi}{\psi}=1_{\cB}$ with
$p_\psi := \lhs{\psi}{\psi}\in\cA$ the corresponding projection.
Let $\onabla$ be the anti-holomorphic connection on $\cE$.
Then, the projection $p_\psi$ is a solution of the self-duality equations of \eqref{sd-asd}, 
\beq\label{sdbis0}
\overline{\pa} (p_\psi) ~p_\psi = 0,
\eeq
if and only if the $\psi$ is a generalized eigenvector of $\onabla$, i.e. there exists $\lambda\in \cB$ such that 
\beq\label{sdah0}
\onabla\psi = \psi\lambda\ .
\eeq
\end{proposition}
\begin{proof}
Using the compatibility \eqref{lcomcov} for $\onabla$, and the associativity \eqref{ascond}, one computes:
\begin{align*}
\opa(p_\psi) & =
\opa(\lhs{\psi}{\psi}) = \lhs{\onabla\psi}{\psi} + \lhs{\psi}{\nabla \psi} \,, \\
\text{and} \qquad 
\opa(p_\psi) p_\psi
&= \lhs{\onabla\psi}{\psi}\lhs{\psi}{\psi} + \lhs{\psi}{\nabla \psi}\lhs{\psi}{\psi} \\
& = \lhs{ \lhs{\onabla\psi}{\psi} \psi}{\psi} + \lhs{\lhs{\psi}{\nabla \psi} \psi}{\psi}  \\
& = \lhs{ \onabla\psi \rhs{\psi}{\psi}} {\psi} + \lhs{\psi \rhs{\nabla \psi}{\psi}} {\psi} \\
& = \lhs{ \onabla\psi  + \psi \rhs{\nabla \psi}{\psi}} {\psi} .
\end{align*}
Now, the operator $\onabla$ on $\rhs{\psi}{\psi}=1_{\cB}$ and the right compatibility \eqref{rcomcov} yields:
$$
\rhs{\nabla \psi}{\psi} = - \rhs{\psi}{\onabla\psi} ,
$$
which, when inserting in the previous expression leads to:
$$
\opa(p_\psi) p_\psi = \lhs{\onabla\psi - \psi \rhs{\psi}{\onabla\psi}} {\psi} ,
$$
that is $\opa(p_\psi) p_\psi = 0$ if and only if $\lhs{\onabla\psi - \psi \rhs{\psi}{\onabla\psi}} {\psi}=0$.
Then, by applying the latter expression to $\psi$ and using the associativity condition \eqref{ascond} and 
$\rhs{\psi}{\psi}=1_{\cB}$, one gets:
$$
\lhs{\onabla\psi - \psi \rhs{\psi}{\onabla\psi}} {\psi} \psi = \left( \onabla\psi - \psi \rhs{\psi}{\onabla\psi} \right) \rhs{\psi}{\psi} 
= \onabla\psi - \psi \rhs{\psi}{\onabla\psi}.
$$
Since $\rhs{\psi}{\onabla\psi}\in\cB$ the equivalence of the two equations \eqref{sdah0} and \eqref{sdbis0} then follows.
\end{proof}

In the present paper we shall present solutions of \eqref{sdah0} and of \eqref{sdbis0} both on the Moyal plane and on noncommutative tori. 

\section{The Schr\"odinger representation}\label{se:sr}

For the construction of equivalence bimodules for the Moyal and the noncommutative torus we are going to rely on 
some basic facts of the Schr\"odinger representation of the phase space $\RRt$ on the Hilbert space $\LtR$ with (left linear) scalar product
$$
\langle \xi, \eta\rangle =\int_{\RR}\xi(t)\, \overline{\eta}(t) \dd t ,
$$
and corresponding norm denoted $\|\cdot\|_2$. 
The Schr\"odinger representation is the projective representation of $\RRt$ on $\LtR$ defined by
\beq\label{sch-rep}
   (\pi(z)\, \xi)(t)=e^{2\pi \mi t\omega}\xi(t-x), \qquad \text{for}~~z=(x,\omega).
\eeq
With $z=(x,\omega)$ and $z'=(x^\prime,\omega^\prime)$ we have then
\beq\label{projrep}
 \pi(z)\,\pi(z^\prime)=e^{-2\pi \mi x\omega^\prime}\pi(z+z^\prime).
\eeq
The map $c:\RR\times\RR\to\TT$ defined by 
\beq\label{csr}
c(z,z^\prime)=e^{-2\pi \mi (x\omega^\prime)}
\eeq
is a 2-cocycle. Another application of \eqref{projrep} gives a commutation relation:
$$ 
 \pi(z)\pi(z^\prime)=c(z,z^\prime)\overline{c(z^\prime,z)}\pi(z^\prime)\pi(z),
$$
which relies on the anti-symmetrized 2-cocycle 
$c_{\mathrm{symp}}(z,z^\prime)=c(z,z^\prime)\overline{c(z^\prime,z)}$. 
Note that $c_{\mathrm{symp}}$ comes from the 
standard symplectic form $\Omega(z,z^\prime)=y\omega-x\eta$ on $\RRt$, hence the name.
 
The matrix-coefficients of the Schr\"odinger representation are defined for $\xi,\eta \in L^2(\RR)$ by
\begin{align} \label{mcsr}
 V_\eta\xi(z):=\langle \xi, \pi(z)\eta\rangle &=\int_{\RR}\xi(t)\overline{\eta}(t-x)e^{-2\pi \mi t\omega}\dd t \nn \\
 & =e^{-\pi \mi x\omega}\int_\RR \xi(t+\tfrac{1}{2}x)\overline{\eta}(t-\tfrac{1}{2}x)e^{-2\pi \mi t\omega} \dd t .
\end{align}
In signal analysis $V_\eta\xi$ is known as the short time Fourier transform. 
The matrix-coefficients of the Schr\"odinger representation are elements of $L^2(\RR^2)$.
They have several important properties that 
we state as Lemmas referring to \cite{fest98} for details and proofs. Firstly, a
basic fact about them is an orthogonality relation:
\begin{lemma}(Moyal identity).
Suppose $\xi,\eta,\varphi,\psi$ are in $L^2(\RR)$. Then, it holds that
\beq\label{moid}
  \langle V_\eta\xi,V_\psi\varphi\rangle_{L^2(\RR^2)} =\langle \xi,\varphi\rangle \overline{\langle \eta,\psi\rangle} ,
  \eeq
with scalar product on the left-hand side the left linear one, with measure $\dd z = \dd x\,\dd\omega$.
 
\end{lemma}
This identity shows that for $\eta$ normalized, $\|\eta\|_2=1$, the map $\xi\mapsto V_\eta\xi$ is an isometry:
\beq
\iint_{\RRt}|V_\eta\xi(x,\omega)|^2\dd z = \|\xi\|_2^2 .
\eeq
It also shows that the Schr\"odinger representation is irreducible on $\LtR$. 

An additional important consequence of the Moyal identity
is a reconstruction formula for $\xi\in\LtR$ in terms of $\{\pi(z)\eta: z\in\RRt\}$.
\begin{lemma}
Let $\eta$ and $\psi$ be in $\LtR$ such that $\langle \psi, \eta \rangle\ne 0$. Then for any $\xi\in L^2(\RR)$, 
$$ 
\xi=\langle \psi,\eta \rangle^{-1} \iint_{\RRt}\langle \xi,\pi(z)\eta\rangle\pi(z)\psi\, \dd z = \langle \psi,\eta\rangle^{-1} \iint_{\RRt}  V_\eta \xi(z) \pi(z)\psi\, \dd z .
$$
\end{lemma}
Finally, the matrix-coefficients of the Schr\"odinger representation of Schwartz elements are them-self Schwartz elements.
\begin{lemma}\label{schwartz}
If $\xi,\eta$ are of Schwartz class $\SR$, then $V_\eta\xi$ is in the Schwartz class $\SRt$.
\end{lemma}

Next, let us consider the twisted group algebra $L^1(\RRt,c)$ of $\RRt$ associated to the cocycle in \eqref{csr},  $c(z,z^\prime)=e^{-2\pi \mi (x\omega^\prime)}$ for $z=(x,\omega)$ and $z'=(x^\prime,\omega^\prime)$. Then $L^1(\RRt,c)$ is an involutive Banach algebra with respect to twisted convolution in $L^1(\RRt)$: that is for $k$ and $l$ in $L^1(\RRt)$, one defines 
the twisted convolution $(k\natural l)$ by  
\beq\label{tcon2}
  (k\natural l)(z)=\iint k(z^\prime)l(z-z^\prime)c(z^\prime,z-z^\prime)\, \dd z^\prime
\eeq
and twisted involution of $k\in L^1(\RRt)$ as
\beq\label{twinv} 
 k^\star(z)=c(z,z)\overline{k(-z)}=e^{-2\pi \mi x \omega}\overline{k(-z)} .
\eeq

The integrated representation
\beq\label{intrep}
 K=\pi(k)=\iint_\RRt k(z)\pi(z) \dd z , \quad \mbox{for} \quad k\in L^1(\mathbb{R}^2) , 
\eeq
is a non-degenerate bounded representation of the twisted convolution algebra $L^1(\RRt,c)$ on $L^2(\RRt)$. The adjoint of $K=\pi(k)$ is given by 
$K^*=\pi(k^\star)$ and the composition of 
$K=\pi(k)$ and $L=\pi(l)$ corresponds to the element $(k\natural l)$:
\beq\label{tcon1}
 K L = \iint_\RRt (k\natural l)(z)\pi(z) \dd z .
\eeq

Operators of the form arising from the integrated representation of the Schr\"odinger representation are pseudo-differential operators when $k\in L^1(\RRt)$, see \cite{FS10} for instances of the vast literature on this approach.
With the twisted involution in \eqref{twinv}, one has that 
\beq\label{timsr}
(V_\eta\xi)^\star=V_\xi\eta
\eeq
in $L^1(\RRt,c)$. Indeed, by the properties of the projective representation $\pi$ it follows that:
$$
(V_\eta \xi)^\star(z):=e^{-2\pi \mi x\omega} \overline{V_{\eta}\xi(-z)}=V_\xi\eta(z) .
$$

We need twisted convolutions of matrix coefficients of the Schr\"odinger representation.
\begin{lemma}
 Let $\xi_1,\xi_2,\eta_1,\eta_2$ be in $\LtR$. Then, it holds that 
\begin{equation} \label{twmc}
  V_{\eta_2}\xi_2\natural V_{\eta_1}\xi_1(z)=\langle \xi_1,\eta_2\rangle V_{\eta_1}\xi_2(z) .
\end{equation}
\end{lemma}
\begin{proof}
Note that $\langle\pi(z^\prime)\xi,\pi(z)\eta\rangle=c(z,z-z^\prime)\langle \xi,\pi(z-z^\prime)\eta\rangle$, which allows one to express the twisted convolution of 
$V_{\eta_2}\xi_2$ and $V_{\eta_1}\xi_1$ as follows:
\begin{align*}
V_{\eta_2}\xi_2\natural V_{\eta_1}\xi_1(z) & = \int_{\RRt}V_{\eta_2}\xi_2(z^\prime)V_{\eta_1}\xi_1(z-z^\prime)c(z^\prime,z-z^\prime)dz \\
&=\int_{\RRt}\langle \xi_2,\pi(z^\prime)\eta_2\rangle\langle\pi(z^\prime)\xi_1,\pi(z)
	\eta_1 \rangle dz\\
&=\langle \xi_2,\pi(x,\omega)\eta_1\rangle\langle \xi_1,\eta_2\rangle\\
&=\langle \xi_1,\eta_2\rangle V_{\eta_1}\xi_2(z) , 
\end{align*}
which is just the stated equality.
\end{proof}

On elements of $L^1(\mathbb{R}^2,c)$ which are of Schwartz class the evaluation at zero, 
\beq\label{twtra}
\tr(k):=k(0) \qquad \text{for} \qquad k\in \SR ,
\eeq
yields a (faithful) trace: $\tr(k\natural l) = \tr(l\natural k)$, for all $k,l\in \SR$. 

\begin{lemma}
Suppose $\xi\in L^2(\RR)$ with $\|\xi\|_2=1$. Then the element $p_{\xi}= V_\xi \xi$ is a projection. 
If in particular $\xi\in \SR$, this projection is in $\SRt$ and has trace $1$.
\end{lemma}
\begin{proof} 
By \eqref{twmc} we have that
\beq
 V_\xi\xi\,\natural\, V_\xi\xi=V_\xi\xi .
\eeq
and by \eqref{timsr} it holds that 
\beq
(V_\xi\xi)^\star =V_\xi\xi . 
\eeq
That $\tr(p_\xi) \in \SRt$ for $\xi\in \SR$ follows from Lemma~\ref{schwartz}, and then by the definition of the trace: 
$\tr(p_\xi)=(V_\xi\xi)(0)=\|\xi\|^2=1$.
\end{proof}

We denote by $\cA$ the class of all operators of the form \eqref{intrep} for $k\in\SRt$. 
They are all trace-class \cite{FS10}. Consequently, $\cA$ is a subalgebra  of the 
the algebra of all compact operators $\cK$ on $\LtR$ and its norm closure coincides with $\cK$.
It follows from results on the Sj\"ostrand class in \cite{FS10} and the description of $\cS(\RRt)$ as intersections of weighted Sj\"ostrand classes \cite[Rem.~1.3 (5)]{to07}, that $\cA$ is an inverse-closed subalgebra of $\mathcal{B}(L^2(\RR))$. This is a result needed for the Schwartz space $\cS(\RR)$ to be a projective module over the smooth algebra $\cA$. 

\section{The Moyal plane and its geometry}\label{se:moyp}

We view the algebra $\cA$ of previous section, that is all operators as in \eqref{intrep} for $k\in\SRt$, as a model of the Moyal plane for the following reason. The product of operators in $\cA$ was defined in terms of 
the twisted convolution on $L^1(\RRt)$, but if one considers the Fourier transform of the symbols defining the elements in $\cA$ one obtains the Moyal product:
\beq
k \star l=\mathcal{F}^{-1}\big(\mathcal{F}(k) \natural \mathcal{F}(l) \big)  \qquad \text{for} \quad k,l \in \cS(\RRt) .
\eeq

An important result of Rieffel in \cite{Ri72} shows that the uniqueness of the Heisenberg commutation relations is equivalent to the Morita equivalence between $\mathcal{K}$ and the complex numbers $\CC$. 
In this section we establish that, as a consequence, the space $\SR$ is a smooth
$\cA-\CC$ equivalence bimodule. Later on we equip this equivalence bimodule with connections that are compatible 
with derivations on the Moyal plane algebra $\cA$.

\begin{proposition}\label{bimoyal}
The space $\cE=\SR$ is an equivalence bimodule between $\cA$ and $\CC$ with respect to the actions:   
\beq\label{lhs}
K\cdot \xi =\iint k(z)\pi(z)\xi\, \dd z , \quad \textup{and} \quad  \xi \cdot\lambda =\xi\,\overline{\lambda}  , 
\eeq
for $\xi\in\SR$ and $k\in\SRt$, $\lambda\in\CC$; and $\cA$-valued and $\CC$-valued hermitian products: 
\beq\label{lact}
\lhs{\xi}{\eta} = \iint\langle \xi,\pi(z)\eta\rangle\pi(z) \dd z = \iint V_\eta\xi(z) \pi(z) \dd z \quad \textup{and} \quad
\rhs{\xi}{\eta} = \langle \eta,\xi\rangle,  
\eeq
for $\xi,\eta\in\SR$. Here $\langle \cdot, \cdot \rangle = \langle \cdot, \cdot \rangle_{\LtR}$ is the scalar product on $\LtR$.
\end{proposition}
\noindent
Note that $\lhs{\xi}{\eta} $ is noting but the operator $\pi(V_\eta \xi)$.
\begin{proof}
By the invariance of $\SR$ under $\pi(z)$ for all $z\in\RRt$, one obtains that $K\cdot\xi \in \SR$. 
Lemma \ref{schwartz} states that for $\xi,\eta\in\SR$ also $V_\eta\xi\in\SRt$ and hence $\lhs{\xi}{\eta}$ is 
in $\cA$. The statements for $\rhs{\cdot}{\cdot}$ and the right action are the definition for complex-conjugate 
in $\LtR$.

\noindent
It remains to show the associativity condition:
\beq\label{morcom}
\lhs{\xi}{\eta} \cdot \psi = \xi\cdot \rhs{\eta}{\psi} 
\eeq
Note that, after taking an $\LtR$ scalar product with an additional function $\varphi$, 
this is equivalent to Moyal identity \eqref{moid}:
\beq
 \langle \lhs{\xi}{\eta} \cdot \psi, \varphi\rangle=\langle \xi \cdot \rhs{\eta}{\psi}, \varphi\rangle \qquad \textup{or} \qquad
 \langle V_\eta\xi,V_\varphi \psi\rangle_{L^2(\RR^2)}=\langle \xi,\varphi\rangle\langle \psi,\eta\rangle .
\eeq
Consequently, by completing $\SR$ with respect to the left structure $\lhs{\cdot}{\cdot}$ one gets an equivalence bimodule 
$E$ between the compact operators $\mathcal{K}$ and $\CC$. By the inverse-closedness of $\cA$ in $I+\cK$ we deduce that 
$\cE=\SR$ is an equivalence bimodule between $\cA$ and $\CC$.
\end{proof}

\subsection{The derivations}
A two-dimensional noncommutative geometry, that is a $2$-summable spectral triple for the Moyal plane is given as follows.
Firstly, there are derivations (that is an infinitesimal action of $\IT^2$) $\partial_1$ and $\partial_2$ on
$\cA$ defined by:
\begin{align} \label{2der}
  \partial_1 K&=2\pi \mi \iint_\RRt x k(x,\omega)\pi(x,\omega)\, \dd x \dd \omega , \nn \\
  \partial_2 K&=2\pi \mi \iint_\RRt \omega k(x,\omega)\pi(x,\omega)\, \dd x \dd \omega .
\end{align}
Since $k\in\SRt$, so are $x k$ and $\omega k$; thus $\pa_j(K)\in \cA$ for $j=1,2$. Then,
with the rule for the product given in \eqref{tcon1} and \eqref{tcon2}, one directly checks that these are derivations, that is
$\pa_j(K L) = \pa_j(K) L + K \pa_j(L)$ for $j=1,2$. They clearly commute with each other.

For ease of notation, let us write $\ac = \cA + \IC 1$ for the unitization. The derivations $\partial_j$ 
are naturally extended to all of $\ac$ with the trivial action on constants.

The GNS representation space $\cH_0$ is defined as the
completion of $\ac$ for the Hilbert norm $||K||_{\rm GNS}:=\sqrt{\tr(k^\star \,\natural k )}=\|k\|_{2}$
with $k\in\ac$ and trace given in \eqref{twtra} and extended to be zero on scalars.
The space $\cH_0$ is unitarily equivalent to the completion of $\SRt$ with the norm $\sqrt{(k^*  \natural \, k)(0)}$.
As usual, the algebra $\ac$ is represented
faithfully as bounded operators on $\cH_0$:
\beq\label{rep0}
\pi(K) \widehat{L}= \widehat{K \,L} ,
\eeq
for any $K,L\in\ac$. On the Hilbert space $\cH=\cH_0\otimes \IC^2$ the algebra $\ac$ acts diagonally with two copies of the representation $\pi$ in \eqref{rep0}. We consider the Dirac operator $D$ given by
\beq
D = \begin{pmatrix}
0 &  \opa \\
\pa & 0 \end{pmatrix} =
\begin{pmatrix}
0 &  \pa_1 + \tau \pa_2 \\
 \pa_1 +\bar{\tau} \pa_2 & 0
\end{pmatrix} ,
\eeq
with derivations $\pa_1$ and $\pa_2$ given in \eqref{2der}.
For the particular choice $\tau=\mi$ this reduces to $D=\pa_1 \sigma_1 + \pa_2 \sigma_2$, with $\sigma_1,\sigma_2$,
two Pauli matrices. Let us restrict to $\tau=\mi$ from now on.

Firstly, $D$ is self-adjoint on a natural dense domain in $\cH$.
Moreover, for $K \in \ac$ the commutator $[D, K]$ is bounded since it is just multiplication by the functions $\opa(K)$ and $\pa(K)$. Also, $D^2 = \Delta \otimes \IC^2$ with the Laplacian $\Delta=\opa \pa = \pa_1^2 + \pa_2^2$ acting as:
$$ 
  \Delta K = - 4 \pi^2 \iint_\RRt (x^2 + \omega^2) k(x,\omega)\pi(x,\omega)\, \dd x \dd \omega .
$$
There is in fact more structure. A grading operator is:
\beq
\gamma =
\begin{pmatrix}
1 & 0 \\
0 & -1
\end{pmatrix} .
\eeq
Furthermore, the vector $1 = \widehat{I}$ of $\cH_0$ is cyclic (i.e. $\pi(\ac) 1$ is dense in $\cH_0$) and
separating (i.e. $\pi(K) 1 = 0$ implies $K=0$). Also, the state $(1, \pi(K) 1)$ is evidently tracial.
Thus, the Tomita involution on $\cH_0$ is just
\beq
J_0(\widehat{K}) = \widehat{K^\star} , \qquad \textup{for all} \quad \widehat{K}\in\cH_0 .
\eeq
On the Hilbert space $\cH$ the reality operator is:
\beq
J = \begin{pmatrix}
0 & -J_0 \\
J_0 & 0
\end{pmatrix} .
\eeq
The datum $(\ac, \cH, D, \gamma, J)$ was already considered in \cite{GGISV} (albeit in `two-dimensional position' space) where it was shown that it constitutes a real even 2-summable spectral triple.

\subsection{The constant curvature connection}\label{se:como}
On the equivalence bimodule $\cE=\SR$ one defines a connection via covariant derivatives $\nabla_1$ and $\nabla_2$ given by:
\beq\label{ccc}
(\nabla_1\xi)(t)=2\pi \mi\, t\,\xi(t), \qquad \textup{and} \qquad (\nabla_2\xi)(t)=\xi^\prime(t).
\eeq
It is a direct computation to check that they obey:
\begin{align*}
  \nabla_1(K\cdot \xi)(t) &= 2\pi \mi \iint k(z)[ t \pi(z)\xi (t)] \, \dd z \\
  &= 2\pi \mi \iint x k(z)\pi(z)\xi(t) \, \dd z + 2\pi \mi \iint k(z)\pi(z) [t \xi(t)] \, \dd z.
\end{align*}
and
\begin{align*}
  \nabla_2(K\cdot \xi)=\iint k(z)[\pi(z)\xi]^\prime \, \dd z =2\pi \mi \iint \omega k(z)\pi(z)\xi\, \dd z +\iint k(z)\pi(z)\xi^\prime \, \dd z.
\end{align*}
With the derivations in \eqref{2der}, these are just the left Leibniz rule for the connection:
\beq
\nabla_j (K\cdot \xi)=(\partial_j K)\cdot \xi+K\cdot(\nabla_j \xi), \qquad j =1,2 .
\eeq

The covariant derivatives are compatible with the (left) hermitian structure:
$$ 
\partial_j (\lhs{\xi}{\eta})=\lhs{\nabla_j \xi}{\eta} + \lhs{\xi}{\nabla_j \eta}, \qquad j =  1,2 .
$$
In view of the second equality in \eqref{lact}, these reduce to statements on the corresponding integrands.
For $j=1$ this is:
$$
  x V_\eta \xi (x,\omega) =V_\eta \tilde{\xi} (x,\omega) +V_{\tilde{\eta}} \xi (x,\omega) ,
$$
with $\tilde{\xi}(t) = t \xi(t)$ and similarly for $\tilde{\eta}$. When $j=2$ it is instead:
\beq
  2\pi \mi \, \omega V_\eta\xi(x,\omega)=V_\eta\xi^\prime(x,\omega)+V_{\eta^\prime}\xi(x,\omega) ,
\eeq
both of which can be established by direct computations. 

On the other hand, property of the $\LtR$ scalar product leads directly to
the compatibility with respect to the right hermitian structure expressed as:
$$ 
\rhs{\nabla_j \xi}{\eta} + \rhs{\xi}{\nabla_j \eta} = 0, \qquad j =  1,2 .
$$
Since the right algebra $B$ is just $\IC$, the right Leibniz rule for the connection is automatic with respect 
to the trivial (null) derivations.

Finally we observe that the connection has constant curvature:
\beq\label{curvm}
F_{1,2} := [\nabla_1, \nabla_2] = - 2 \pi \mi \,\, \mathrm{id_\cE} .
\eeq

Of course these are none other than the Heisenberg commutation relations (in the Schr\"odinger representation).
And the anti-holomorphic connection $\onabla=\nabla_1 + \mi \nabla_2$ is the annihilation operator while the 
holomorphic connection $\nabla=\nabla_1 - \mi \nabla_2$ is the creation one.

\subsection{Solitons}\label{sol-moy}
As an application of Lemma~\ref{Rieffel}, the equivalence bimodule $\cE=\SR$ allows one to construct projections in the Moyal plane algebra. 
Since the `right' algebra $\cB$ is just complex numbers $\IC$, 
and the `right' hermitian structure is just the usual one, by the bimodule property \eqref{morcom} any $\psi \in \SR $ normalized as $\rhs{\psi}{\psi}=\|\psi\|_2=1$, provides a non-trivial projection $p_\psi = \lhs{\psi}{\psi}$ in 
$\cA$.
Then, we seek solutions of the self-dual equations \eqref{sd-asd} for projections $p_\psi := \lhs{\psi}{\psi}$ in 
$\cA$ with complex derivations $\opa = \pa_1 + \mi \pa_2$ and $\pa = \pa_1 - \mi \pa_2$ (and derivations $\pa_1,\pa_2$ in \eqref{2der}).
From Proposition~\ref{main0}
the projection $p_\psi$ is a solution of the self-duality equation,
\beq\label{sdbism}
\overline{\pa} (p_\psi) ~p_\psi = 0.
\eeq
if and only if the element $\psi$ satisfies 
\beq\label{sdah}
\onabla\psi = \psi \lambda ,
\eeq
for some $\lambda\in  \mathbb{C}$, where
 $\onabla=\nabla_1 + \mi \nabla_2$ is 
 the anti-holomorphic connection and $\nabla_1$ and $\nabla_2$
the covariant derivatives in \eqref{ccc}. 
As mentioned, these are nothing but equations for eigenfunctions of $\onabla$,
all solutions of which are given by the generalized Gaussians 
$$ 
\psi (t) = c \, e^{- \pi t^2 -\mi \lambda t} , \quad \mbox{for} \quad \lambda\in \mathbb{C} , 
$$
where $c$ is a normalization constant.
Then the projection 
$$ 
p_\psi=\lhs{\psi}{\psi}=\iint_{\RRt}  V_\psi \psi(z) \pi(z)\, \dd z 
$$
 solves the self-duality equations \eqref{sdbism}. 
Here $V_\psi\psi$ reads explicitly
$$
V_\psi\psi(x,\omega)=e^{- \frac{\pi}{2}(x^2+\omega^2)}e^{-\pi \mi x\omega
 - \frac{\mi}{2}(\bar\lambda+\lambda)x + \frac{1}{2}(\bar\lambda-\lambda) \omega} , 
$$ 
 and $\tr(p_\psi) = V_\psi \psi (0) = 1$.
Also, Lemma~\ref{tcconst}, with \eqref{curvm} gives for its topological charge:
$$
c_1(p_\psi) = 1  . 
$$
The self-duality equation for these projections is the equation for the minimizers of the Heisenberg uncertainty relations, 
which explains why they are Gaussian $\psi_\lambda$.

\section{The Noncommutative torus}\label{se:nct}
We shall now move to the two-dimensional noncommutative torus $\TT^2_\theta$. A wide class 
of modules over noncommutative tori, called Heisenberg modules, were constructed in \cite{CR87,ri88}. 
For an irrational $\theta$ these are all modules which are not free.  A link between Heisenberg modules 
and Gabor frames was exhibited in \cite{lu09}. 
In view of this, we describe the noncommutative torus via (a restriction of) the Schr\"odinger representation of $\RRt$
in \S\ref{se:sr}. 

\subsection{The torus and its geometry}
Thus, for a non-zero $\theta\in\RR$, at the $C^*$-algebra level the noncommutative torus 
$A_\theta$ is the norm closure of the span of $\{\pithkl : k,l\in\ZZ\}$, with $\pi$ given in \eqref{sch-rep} for $z=(\theta k, l)$. The operators $\pithkl$ provide a (reducible) projective faithful representation on $\LtR$ of the lattice $\theta\ZZ\times\ZZ \subset \RRt$, which 
is just the restriction of the Schr\"odinger representation of $\RRt$.  
Denoting the operators $\pi(0,1)$ and $\pi(\theta,0)$  as $M_1$ and $T_\theta$,
respectively, they satisfy the noncommutative torus relation \cite{Ri81}:
\beq
M_1T_\theta=e^{2\pi i\theta}T_\theta M_1 . 
\eeq

The smooth noncommutative torus is the subalgebra $\cA_\theta$ of $A_\theta$ consisting of operators
\begin{equation}\label{smoothA}
  \pi({\bf a})=\sum_{k,l\in\ZZ}a_{kl}\pithkl, \qquad \textup{for} \quad {\bf a}=(a_{kl})\in\cS(\ZZt).
\end{equation}
Furthermore, with ${\bf a}$ and ${\bf b}$ in $\cS(\ZZt)$ we have for their product
\beq
   \pi({\bf a})\pi({\bf b})=\pi({\bf a}\natural{\bf b})
\eeq
where ${\bf a}\natural{\bf b}$ is the twisted convolution 
\begin{equation*}
  ({\bf a}\natural{\bf b})(k,l)=\sum_{m,n\in\ZZ}a_{mn}b_{k-m,n-l}e^{-2\pi i\theta n(k-m)}
\end{equation*}
while $\pi({\bf a})^*=\pi({\bf a}^*)$, where ${\bf a}^*$ is the twisted involution of ${\bf a}$:
\begin{equation*}	
  (a_{kl})^*=e^{-2\pi i\theta kl}\overline{a_{-k,-l}} .
\end{equation*}

Operators commuting with $\pithkl$ for all $k,l\in\ZZ$ are those associated with 
the lattice $\ZZ\times\theta^{-1}\ZZ$, since, by the commutation relations for the Schr\"odinger representation,  
$$ 
 \pi(z)\pithkl=\pithkl\pi(z) ,
$$
holds if and only if 
$c_{\mathrm{symp}}(z,(\theta k,l))=1$ for all $k,l\in\ZZ$, with the anti-symmetrized 2-cocycle $c_{\mathrm{symp}}$ of \S\ref{se:sr}.
The norm closure of the span of 
$\{\pi(k,\theta^{-1}l) : k,l\in\ZZ\}$ is the noncommutative torus $A_{1/\theta}$ 
and the operators $T_{1}$ and $M_{1/\theta}$ 
provide a faithful representation of $A_{1/\theta}$ on $\LtR$. In parallel with \eqref{smoothA}, the smooth algebra $\cA_{1/\theta}$ is now made of elements
\begin{equation}\label{smoothB}
  b=\sum_{k,l\in\ZZ}b_{kl} \pikthl , \qquad \textup{for} \quad {\bf b}=(b_{kl})\in\cS(\ZZt) \, 
\end{equation}
For the construction of projections in $\cA_\theta$ we will use an equivalence bimodule between $\cA_\theta$ and an algebra $\cB$.
Since this requires $\cB$ to act from the right, instead of $\cA_{1/\theta}$ we have to use for $\cB$ the 
opposite algebra $(\cA_{1/\theta})^{\mathrm{op}}$ which we identify with $\cA_{-1/\theta}$.

An equivalence bimodule between the smooth algebras $\cA=\cA_\theta$ and $\cB=\cA_{-1/\theta}$ is given once more by $\cE = \SR$ \cite{ri88} (see also \cite{lu09}). We state this result as a lattice version of Proposition~\ref{bimoyal} for the Moyal plane.
\begin{proposition}\label{binctl}
The space $\cE=\SR$ is an equivalence bimodule between the smooth noncommutative tori $\cA$ and $\cB$ 
with respect to the actions: 
\begin{align}
a\cdot\xi =\sum_{k,l\in\ZZ}a_{kl}\pithkl \xi , \quad \textup{and} \quad \xi\cdot b =\sum_{k,l\in\ZZ}b_{kl}\pikthl^* \xi,
\end{align}
for $a\in \cA, b\in \cB$ and $\xi,\eta\in\SR$; and with hermitian products given, for $\xi,\eta\in\SR$, by
\begin{align}
		\lhs{\xi}{\eta} & =\theta \sum_{k,l\in\ZZ}V_\eta\xi(\theta k,l)\pithkl , \nonumber \\ \quad \textup{and} \quad
		\rhs{\xi}{\eta} & =\sum_{k,l\in\ZZ} V_{\xi}\eta(k,{l}{\theta^{-1}})\pikthl . 
\end{align}
\end{proposition}
\noindent
From the expression \eqref{mcsr}, the coefficients of the lattice Schr\"odinger representation become
\beq
V_\eta\xi(\theta k,l) = \int_\RR \xi(t)\overline{\eta(t-\theta k)}e^{-2\pi\theta lt}dt
\eeq
and with a similar formula for the coefficients in the right hermitian product.

On the left algebra $\cA_\theta$ we shall use the (normalized) trace $\tr({\pi(\bf a}))= a_{00}$. Then, the compatibility 
$\tr \; \lhs{\xi}{\eta} =  \tr \rhs{\eta}{\xi}$, for all $\xi, \eta \in \cE$, requires that the trace on the right algebra 
$\cA_{-1/\theta}$ be the (non-normalized) one given by $\tr({\pi(\bf b)})=\theta\,b_{00}$.  

The infinitesimal action of an undeformed torus $\IT^2$ on both algebras $\cA_\theta$ and $\cA_{-1/\theta}$, are derivations. 
For lattice versions of the noncommutative tori the derivations on $\cA_\theta$ are just
\begin{align}\label{2dernct}
\pa_1(a) & = 2\pi \mi \, 
\sum_{k,l}k a_{k,l}\pithkl \, \nonumber \\ \quad \textup{and} \qquad  
\pa_2(a) & =2\pi \mi \, 
\sum_{k,l}l a_{k,l}\pithkl , 
\end{align}
and the ones on $\cA_{-1/\theta}$ are then
\begin{align}\label{derivdual}
\pa_1(b) &= -2\pi \mi \theta^{-1}\sum_{k,l}k b_{k,l}\pikthl^* \, \nonumber \\ \quad \textup{and} \qquad 
\pa_2(b) &= -2\pi \mi \theta^{-1} \sum_{k,l}l b_{k,l}\pikthl^* .
\end{align}
The derivations lift to covariant derivatives $\nabla_1$, $\nabla_2$ on the equivalence bimodules $\cE=\SR$ which are given by:
\beq\label{ccctorus}
  (\nabla_1\xi)(t)=2\pi \mi \, \theta^{-1} \, t \, \xi(t) \qquad \textup{and} \qquad (\nabla_2\xi)(t) = \frac{\dd \xi (t)}{\dd t} =: \xi^\prime(t) .
\eeq
The covariant derivatives satisfy
\begin{align*}
\nabla_1(a\cdot \xi)(t) &= 2\pi \mi \theta^{-1} \sum_{k,l} a_{kl}\, t \pithkl\xi (t)\\
                          &= 2\pi \mi 
\sum_{k,l} ka_{kl}\,\pithkl\xi(t) + 2\pi\mi \theta^{-1}  \sum_{k,l} ka_{kl}\,\pithkl [t \xi(t)]    
\end{align*}
and
\begin{align*}
  \nabla_2(a\cdot \xi)(t)
  =\sum_{k,l} a_{kl}\,[\pithkl\xi]^\prime(t)
  =2\pi\mi\sum_{k,l} l a_{kl}\,\pithkl\xi(t) +  \sum_{k,l} a_{kl}\,\pithkl\xi^\prime(t).
\end{align*}
With the derivations in \eqref{2dernct}, these are just the left Leibniz rule for the connection:
\beq
\nabla_j (a\cdot \xi)=(\partial_j a)\cdot \xi+a\cdot(\nabla_j \xi), \qquad j =  1,2 .
\eeq

Similarly to the Moyal plane case the covariant derivatives are compatible with the (left) hermitian structure:
$$ 
\partial_j (\lhs{\xi}{\eta})=\lhs{\nabla_j \xi}{\eta} + \lhs{\xi}{\nabla_j \eta}, \qquad j =  1,2 .
$$
In view of the definition of the left hermitian product in Proposition \ref{binctl}, these again reduce to statements on the corresponding integrands.
For $j=1$ this is
$$
  2\pi \mi \, \theta\, k\,  V_\eta \xi (\theta k,l) =V_\eta \tilde{\xi} (\theta k,l) +V_{\tilde{\eta}} \xi (\theta k,l) ,
$$
with $\tilde{\xi}(t) = t \xi(t)$ and similarly for $\tilde{\eta}$. When $j=2$ it is instead:
\beq
  2\pi \mi \, \theta\, V_\eta\xi(\theta k,l)=V_\eta\xi^\prime(\theta k,l)+V_{\eta^\prime}\xi(\theta k,l) ,
\eeq
both of which are the relations of the Moyal case when restricted to the lattice $\thZZ$.
 
In a completely similar manner one establishes the `dual' version of the above statements, that is, the
right Leibniz rule:
\beq
\nabla_j (\xi \cdot b)= (\nabla_j \xi) \cdot b + \xi \cdot (\partial_j b), \qquad j =  1,2 ;
\eeq
and the compatibility with the right hermitian product:
$$ 
\rhs{\nabla_j \xi}{\eta} + \rhs{\xi}{\nabla_j \eta} = \partial_j (\rhs{\xi}{\eta}) , \qquad j =  1,2 .
$$

Finally we observe that the connection has constant curvature:
\beq\label{cocurv}
F_{1,2} := [\nabla_1, \nabla_2] = - 2 \pi\mi \, \theta^{-1} \, \mathrm{id_\cE} , 
\eeq
which acts on the module $\cE=\SR$ on the left.

\subsection{Duality and Gabor frames}\label{dgf}
From Proposition~\ref{binctl}, there exists a standard module Parseval frame $\{\eta_1,...,\eta_n\}$ for $\SR$, that is
each $\xi\in\SR$ has an expansion,  
\begin{equation*} 
\xi=\lhs{\xi}{\eta_1}\eta_1+\cdots+\lhs{\xi}{\eta_n}\eta_n ,
\end{equation*}
and $\SR$, as a module over $\cA_\theta$, is of finite rank and projective. 

It turns out \cite{feka04} 
that for $0<\theta<1$, the module $\SR$ is given by a projection 
in $\cA_\theta$ itself. Thus, for $0<\theta<1$ 
one can use a one-element Parseval frame $\eta$, and 
from $\xi=\lhs{\xi}{\eta}\eta$ for any $\xi\in\SR$, the associativity condition yields that $\rhs{\eta}{\eta}=1$ so that 
$\lhs{\eta}{\eta}$ is a projection in $\cA_\theta$ as in Lemma~\ref{Rieffel}, and moreover, $\SR=\cA_\theta(\lhs{\eta}{\eta})$. 

We know from the discussion in \S\ref{prbim} that from any standard module frame $\eta$ one gets 
a Parseval frame $\tilde{\eta}$ by taking the element $\tilde{\eta}:=\eta(\rhs{\eta}{\eta})^{-1/2}$.
Now, from Proposition~\ref{binctl}, elements of $\cB = \cA_{-1/\theta} \simeq (\cA_{1/\theta})^{\mathrm{op}}$, in particular those of the form $\rhs{\xi}{\eta}$ act `on the left' on $\SR$ via the conjugate representation $\pi^*$ and elements $\pikthl^*$. A standard module frame for the $\cA_\theta$-module $\SR$ has an important property established in \cite{dalala95,ja95,rosh97} which is known as 
{\it duality principle}. Firstly, the $\cA_\theta$-module $\SR$ 
has a single generator $\eta$ if and only if the set $\{\pikthl \eta : k,l\in\ZZ\}$  
is a Riesz basis for its closed linear span $\overline{\mathrm{span}\{\pikthl \eta : k,l\in\ZZ\}}$, 
that is, there exist positive constants $c_1,c_2$ such that 
$$
c_1\sum_{k,l}|a_{k,l}|^2\le\|\sum_{k,l}a_{k,l}\pikthl \eta\|^2\le c_2 \sum_{k,l}|a_{k,l}|^2
$$
for all $(a_{k,l})\in\ell^2(\ZZt)$. This is in parallel with condition \eqref{frame} for the existence of a standard module frame.
As a consequence (cf. \cite{dalala95,ja95,rosh97}), 
with a single left generator $\eta$ and function $\eta^\circ:=\eta\,\rhs{\eta}{\eta}^{-1}$, 
one gets a bi-orthogonal basis $\{\pikthl \eta^\circ : ~k,l\in\ZZ\}$ for 
$\overline{\mathrm{span}\{\pikthl \eta : k,l\in\ZZ\}}$. 
In turn, elements $\xi \in \SR$ have an expansion over $\cB$ as well,
\beq\label{Riesz}
  \xi = \eta\rhs{\eta^\circ}{\xi} ,
\eeq
where $\rhs{\eta^\circ}{\xi}\in\cB$ is uniquely determined. 
In particular, for the Parseval standard module frame $\tilde{\eta}:=\eta(\rhs{\eta}{\eta})^{-1/2}$ 
the {\it duality principle}, also known as {\it Wexler--Raz identity}, reads as an expansion of each $\xi$ in $\SR$ both over $\cA$ and over $\cB$, 
\beq\label{RieszParseval}
  \xi = \lhs{\xi}{\tilde{\eta}} \tilde{\eta} = \tilde{\eta}\rhs{\tilde{\eta}}{\xi} ,
\eeq 
with $\lhs{\xi}{\tilde{\eta}}\in\cA$ and $\rhs{\tilde{\eta}}{\xi}\in\cB$ which are uniquely determined.

\medskip

The link between projective modules over $\cA_\theta$ and signal analysis \cite{lu09,lu11}
provides new classes of projections in $\cA_\theta$ out of Gabor frames. As explicit examples of the latter, we mention here 
Hermite functions and totally positive functions of finite type.

The Hermite functions $\psi_k$ are defined as $\psi_k(t)=c_k e^{\pi t^2}\frac{d^k}{dt^k}e^{-2\pi t^2}$, for $c_k$ an irrelevant normalization constant. 
In the context of Gabor analysis, it was shown in \cite{grly09} that each $\psi_k$ is a standard module frame for $\cS(\RR)$ (as left $\cA_\theta$-module) 
whenever $0<\theta<(k+1)^{-1}$.  

Further examples come from totally positive functions. We recall the following \cite{Sc47,Sc51}:
\begin{definition}\label{tpf}
The function $\eta \in \cS(\RR)$ is said to be \emph{totally positive} if for every two sets of increasing 
real numbers $x_1<\cdots<x_N$ and $y_1<\cdots<y_N$, it holds that
$$
\mathrm{Det} \big(\eta(x_j-y_k)\big)_{1\le j,k\le N} \geq 0 .
$$
Moreover, a totally positive function $\eta$ is of {\it finite type} $M\in\mathbb{N}$ with $M\ge 2$, if its Laplace transform 
$\widehat{\eta}$ has the form
\begin{equation*}
\widehat{\eta}(\omega)=e^{-\delta \omega^2}e^{-\delta_0\omega}\prod_{j=1}^M(1+2\pi i\delta_j\omega)^{-1} ,
\end{equation*}
for real non-zero parameters $\delta_j$, and $\delta>0$. 
\end{definition}
The Gaussian is totally positive. It addition to Gaussians, an example of such a function is $\eta(t)=\cosh(t)^{-1}$.  
Again in the context of signal analysis, it was shown in \cite{grst13} 
that any totally positive function of finite type (greater than $2$) $\eta$ satisfies the condition $\xi=\lhs{\xi}{\eta}\eta$, 
for all $\xi\in\SR$ if and only if $0<\theta<1$. 
Hence, totally positive (Schwartz) functions $\eta$ of finite type are standard module frames for $\SR$ over $\cA_\theta$ precisely for $0<\theta<1$. 

Recall from the discussion of \S\ref{prbim} that from a standard module frame $\eta$ for $\SR$ it is possible to pass to a Parseval frame 
$\tilde{\eta}:=\eta(\rhs{\eta}{\eta})^{-1/2}$. Then, by the  discussion of projections from equivalence bimodules again in 
\S\ref{prbim} we have the following result.
	
\begin{lemma}\label{totposhermproj}
For a standard module frame $\eta$ for $\SR$ over $\cA_\theta$, let $\tilde{\eta}:=\eta(\rhs{\eta}{\eta})^{-1/2}$ be the corresponding Parseval frame. Then:
\item{1.)}  
The Hermite function $\eta=\psi_k$ gives a projection $p_k=\lhs{\tilde\eta}{\tilde \eta} \in \cA_\theta$, 
if $0<\theta<(k+1)^{-1}$. 
\item{2.)} Let $\eta$ be any totally positive function in $\SR$ of finite type greater than $2$ as in Definition~\ref{tpf}. 
Then,   
$p_{\tilde\eta}=\lhs{\tilde\eta}{\tilde \eta}$ is a projection in $\cA_\theta$ if and only if $0<\theta<1$. 
\end{lemma}

All of these projections have topological charge equal to $1$. This follows from Lemma~\ref{tcconst}, with the curvature \eqref{cocurv}, and from   
the unit of $\cA_{-1/\theta}$ having trace equal to $\theta$.

\subsection{Solitons}

It was shown in \cite{DKL00,DKL03} that Gaussian functions are solitons on noncommutative tori, 
in the sense that they solve a self-duality equation. Using results from signal analysis we are able to 
extend these results: any single standard module frame generator $\eta$ for $\SR$ is a soliton on 
noncommutative tori. This means that Gabor frames are solutions 
of the self-duality equation. Thus, by our discussion above, totally positive functions of finite 
type and Hermite functions provide two classes of new solitons that generalize the Gaussian solitons.   

As for the general situation in Proposition~\ref{main0} (and the Moyal case of \S\ref{sol-moy}), 
once again, the projection $p_\psi := \lhs{\psi}{\psi}$ is a solution of the self-duality equations,
\beq\label{sdbismt}
\overline{\pa} (p_\psi) ~p_\psi = 0.
\eeq
if and only if the element $\psi$ satisfies the condition:
\beq\label{sdahmt}
\onabla\psi = \psi \rhs{\psi}{\onabla\psi} ,
\eeq
with $\onabla=\nabla_1 + \mi \nabla_2$ the anti-holomorphic connection and $\nabla_1$ and $\nabla_2$
the connection \eqref{ccctorus}. 

By Lemma \ref{totposhermproj} and the discussion in \S\ref{dgf} leading to the duality property \eqref{RieszParseval}, we have a class of solitons that generalizes the Gaussian solitons in two different ways.
\begin{proposition}
Let $\psi$ be a standard module Parseval frame for the projective $\cA_\theta$-module $\SR$. Then the corresponding projection 
$p_\psi := \lhs{\psi}{\psi}$ satisfies the 
soliton equation:
\beq
  \overline{\pa} (p_\psi) ~p_\psi = 0.  
\eeq 
\end{proposition}
\begin{proof}
From the duality formula \eqref{RieszParseval} applied to $\onabla\psi$, we know that the uniquely determined element 
$\lambda=\rhs{\psi}{\onabla\psi}\in\cA_{-1/\theta}$ is such that $\onabla\psi-\psi\lambda=0$. The equivalence of \eqref{sdbismt} and \eqref{sdahmt} then establishes the result. 
\end{proof}

As in Lemma \ref{totposhermproj}, we know how to pass from a standard module frame $\eta$ for $\SR$ to a Parseval one.
If we denote this as $\psi:=\eta(\rhs{\eta}{\eta})^{-1/2}$, we have the following result: 
\begin{corollary}
With $\psi:=\eta(\rhs{\eta}{\eta})^{-1/2}$, the projection $p_\psi=\lhs{\psi}{\psi}\in\cA_\theta$ satisfies the 
self-duality equations in the following cases:
\item{1.)}  For $0<\theta<(k+1)^{-1}$, if $\eta$ is the k-th Hermite functions $\psi_k$.
\item{2.)}  For $0<\theta<1$, if $\eta$ is a totally positive function in $\SR$ of finite type greater than $2$, as given  in Definition~\ref{tpf}. 
\end{corollary}

In particular, as already shown in \cite{DKL00,DKL03}, the Gaussian function  
$$ 
\psi_\lambda(t) = c\, e^{- \frac{\pi}{\theta} \, t^2 -\mi \lambda t} \,, \qquad\text{for}~~\lambda\in\IC , 
$$
obeys the equation $\onabla \psi_\lambda  = \psi_\lambda \lambda$, and thus solves \eqref{sdahmt}. 
We can confirm now that the right hermitian product 
$\rhs{\psi_\lambda}{\psi_\lambda}$ is indeed invertible in $\cA_{-1/\theta}$ for all $0<\theta<1$ (a question left partially 
open in \cite{DKL00,DKL03}) so that the projections 
$p_\lambda=\lhs{\tilde\psi_{\lambda}}{\tilde\psi_{\lambda}}$, with
$\tilde{\psi_{\lambda}}:=\psi_{\lambda} ( \rhs{\psi_{\lambda}}{\psi_{\lambda}})^{-1/2}$
are solutions of the self-duality equation \eqref{sdbismt}.
The moduli space of such Gaussian solutions was found in \cite{DKL00,DKL03} to be a copy of the complex torus.

The case of projections in noncommutative tori generated by Gaussians and Gabor frames was also discussed in \cite{luma09}, 
where the complex structure was implicitly used. It was pointed out that these projections have extra structure, namely, 
they are quantum theta functions. Consequently, the results in this section yield that quantum theta functions are solutions of the 
self-duality equation and may be considered as a special class of solitons over noncommutative tori, a fact already alluded to in \cite{la06}.

\end{document}